\documentclass[12pt]{article}
\usepackage{arxiv}
\usepackage{amsmath,amssymb,amsfonts,amsthm}
\usepackage[ruled,vlined]{algorithm2e}
\usepackage{graphicx}
\usepackage{float}
\usepackage{xcolor}
\usepackage{enumitem}
\usepackage{url}
\usepackage{tikz}
\usetikzlibrary{arrows.meta, positioning, calc, fit, backgrounds}

\theoremstyle{plain}
\newtheorem{thm}{Theorem}

\newtheorem{prop}[thm]{Proposition}
\newtheorem{defn}{Definition}
\newtheorem{rem}{Remark}
\newtheorem{example}{Example}

\title{\textbf{Explicit Factorization of $X^n-1$ over $\mathbb{Z}_{p^e}$ via Cofactor-Free Single-Seed Hensel Lifting}}

\author{
    Yongchao Wang$^{1}$, Yang Ding$^{2}$, Jiansheng Yang$^{2}$, and Zhiqiu Huang$^{1}$ \\[0.5em]
    $^{1}$\textit{College of Computer Science and Technology, Nanjing University of Aeronautics and Astronautics, Nanjing, China} \\
    \texttt{\{wangyc, zqhuang\}@nuaa.edu.cn} \\[0.5em]
    $^{2}$\textit{Department of Mathematics, Shanghai University, Shanghai 200444, China} \\
    \texttt{dingyang@shu.edu.cn, yjsyjs@staff.shu.edu.cn}
}

\date{\today}

\begin{document}

\maketitle

\begin{abstract}
We present a complete framework for the explicit factorization of $X^n-1$ over integer residue rings $\mathbb{Z}_{p^e}$ for arbitrary $n$ with $\gcd(n, p)=1$. Classical approaches face fundamental bottlenecks: polynomial Hensel lifting requires updating global cofactors (scaling with $n$), while direct multivariate Newton--Hensel iteration on the factor coefficients requires Jacobian inversion (scaling exponentially as $O(p^{(m-1)^2})$ per layer due to zero-divisors, where $m$ is the coset dimension). Our framework eliminates both bottlenecks through three contributions: (1)~the \emph{Ideal Derivation Modulo Principle}, which characterizes all factor coefficients as roots of a multivariate Dickson polynomial ideal derived via modular remainder extraction; (2)~a \emph{cofactor-free Hensel lift} that elevates a single seed factor from $\mathbb{F}_p$ to $\mathbb{Z}_{p^e}$ using a cached polynomial inverse computed once over $\mathbb{F}_p$; and (3)~a \emph{dual-track coefficient reconstruction} mechanism that recovers all remaining factors from the lifted seed's trace array via MED-based coset dispatch, with Newton--Girard inversion as the primary path and quotient-ring Gaussian elimination as an unconditional fallback when $p \leq m$. Empirical evaluation confirms the theoretical grand total algebraic complexity of $O(n + m^3 \log p + e \cdot m^2)$ for explicitly factoring $X^n-1$ over $\mathbb{Z}_{p^e}$, validating the near-constant per-layer lifting cost $O(m^2)$ to depths exceeding $e = 1000$. The framework yields speedups of $445\times$ (including runtime auto-seeding overhead) over SageMath's C-backed FLINT/Pari engine and $33.5\times$ over the V1 scalar lift.
\end{abstract}

\section{Introduction}

The factorization of $X^n-1$ over integer residue rings $\mathbb{Z}_{p^e}$ is a fundamental operation that underpins three pillars of contemporary applied algebra: post-quantum cryptography, fully homomorphic encryption, and algebraic coding theory.

\textbf{Post-Quantum Cryptography.}
In August 2024, NIST finalized its first post-quantum cryptography standards: FIPS~203 (ML-KEM, derived from CRYSTALS-Kyber) and FIPS~204 (ML-DSA, derived from CRYSTALS-Dilithium) \cite{kyber_paper, dilithium_paper}. Both standards operate over quotient rings $\mathbb{Z}_q[X]/(X^n+1)$ with $n=256$ and employ the Number Theoretic Transform (NTT) to reduce polynomial multiplication from $O(n^2)$ to $O(n \log n)$. The NTT requires the modulus $q$ to admit a primitive $2n$-th root of unity, which is equivalent to requiring that $X^{2n}-1$ splits completely modulo $q$. Today's standards restrict $n$ to powers of two and $q$ to a single prime, but extending to composite moduli $q = p^e$ or general cyclotomic orders $n$ would enable richer algebraic structures for next-generation lattice schemes---precisely the regime addressed by our framework.

\textbf{Fully Homomorphic Encryption.}
The BGV and BFV fully homomorphic encryption schemes \cite{brakerski2014leveled, Smart2014Fully} achieve SIMD-style parallel computation on encrypted data through \emph{plaintext slot packing}. The plaintext space $\mathbb{Z}_t[X]/\Phi_m(X)$ decomposes via the Chinese Remainder Theorem into $\ell$ independent ``slots,'' one per irreducible factor of $\Phi_m(X)$ modulo $t$. Maximizing $\ell$ (and hence throughput) demands the complete factorization of $\Phi_m(X)$---a constituent of $X^n-1$---over the plaintext ring. When the plaintext modulus is a prime power $t = p^e$, this factorization must be performed over $\mathbb{Z}_{p^e}$, yet no general-purpose computer algebra system currently supports this operation (as confirmed by our benchmarks in Section~5, where SageMath raises \texttt{NotImplementedError} for all $e > 5$). Our framework fills this gap, enabling slot-packing parameter exploration at arbitrary precision depths.

\textbf{Algebraic Coding Theory.}
Cyclic codes of length $n$ over $\mathbb{Z}_{p^e}$ correspond to ideals in the quotient ring $\mathbb{Z}_{p^e}[X]/\langle X^n-1 \rangle$. The complete irreducible factorization of $X^n-1$ over $\mathbb{Z}_{p^e}$ determines the code's generator polynomial structure: each subset of irreducible factors defines a distinct cyclic code, and the pairwise coprimality of these factors is essential for constructing idempotent generators and for CRT-based encoding and decoding architectures. For finite chain rings, this factorization also governs the Lee and homogeneous distance properties that are critical for codes designed for phase-shift keying (PSK) modulation in modern communication systems. Despite the practical importance, the algebraic tools for performing this factorization efficiently have remained underdeveloped, motivating the present work.

Over finite fields $\mathbb{F}_p$, the factorization of $X^n-1$ is well understood. However, lifting these factorizations to $\mathbb{Z}_{p^e}$ via classical methods encounters severe performance walls. Polynomial Hensel lifting requires maintaining global cofactors of degree $O(n)$, causing high per-layer costs. Alternatively, directly lifting the $m$ coefficients of a factor via multivariate Newton--Hensel requires inverting an $m \times m$ Jacobian matrix at each precision layer. Over $\mathbb{Z}_{p^e}$, zero-divisors obstruct this matrix inversion, forcing exhaustive searches whose cost scales exponentially as $O(p^{(m-1)^2})$ per layer, where $m$ is the coset dimension.

\subsection{Prior Work on Explicit Factorizations}
The explicit determination of irreducible polynomial factors has been extensively explored within finite fields $\mathbb{F}_q$, primarily for individual cyclotomic polynomials $\Phi_d(X)$---which are factors of $X^n-1$ via the decomposition $X^n - 1 = \prod_{d \mid n} \Phi_d(X)$ (see Definition~\ref{def:cyclotomic_cosets}), but are not identical to $X^n-1$ itself.

Fitzgerald and Yucas \cite{Fitzgerald2007} achieved explicit factorizations of the cyclotomic polynomial $\Phi_{2^n \cdot 3}(x)$ over $\mathbb{F}_q$ by projecting solutions onto one-dimensional Dickson scalar symmetries. In our framework, the quadratic irreducible factors arising from $\Phi_{2^n \cdot 3}(x)$ correspond to coset dimension $m = 2$ with norm $P = 1$; the Ideal Derivation Modulo Principle (Theorem~\ref{thm:ideal_modulo}) specializes to the scalar constraint $E_p(S, 1) = 0$, whose roots are precisely the trace values identified in their work.

Wang and Wang \cite{Wang2012} resolved the cyclotomic polynomial $\Phi_{2^n \cdot 5}(x)$ over $\mathbb{F}_q$ through manually derived systems of nonlinear recurrence equations governing the factor coefficients. When Theorem~\ref{thm:ideal_modulo} is evaluated for the coset dimensions arising from $\Phi_{2^n \cdot 5}(x)$ (specifically $m \leq 4$), the remainder coordinates $V_i(\mathbf{A}) = 0$ produce exactly the same coupled polynomial system, confirming that their results are a fixed-dimensional specialization of our general framework.

More recently, Brochero Mart{\'{i}}nez and Ar{\'{e}}valo Baquero \cite{Brochero2019} characterized the irreducible factors of Dickson polynomials directly, though under Kummer extension constraints requiring that every prime divisor of $n$ divides $q-1$.

While these formulations are algebraically significant, they face fundamental limitations. First, they address individual cyclotomic polynomials $\Phi_d(X)$ for specific $d$, relying on low-dimensional symmetry that breaks down for arbitrary orders where mixed-degree coset structures arise. Second, they operate exclusively over fields $\mathbb{F}_q$; lifting their results to integer residue rings $\mathbb{Z}_{p^e}$ immediately incurs the exponential Jacobian barrier associated with multivariate coordinate inversion. Our framework handles the complete factorization of $X^n - 1$ (including all constituent $\Phi_d(X)$ simultaneously) over $\mathbb{Z}_{p^e}$ for arbitrary coprime $(n, p)$.

In our foundational work (V1) \cite{Wang2026ISIT}, we achieved explicit factorization of $X^{p+1}-1$ over $\mathbb{Z}_{p^e}$ by projecting the factorization onto a single Dickson polynomial $V(x)$, enabling Jacobian-free Hensel lifting. However, this derivation was restricted to the specific case $n = p+1$ with coset dimension $m = 2$.

\subsection{Contribution of this Framework}
This paper extends explicit cyclotomic factorization to arbitrary orders $n$ over $\mathbb{Z}_{p^e}$ through three components. First, the Multivariate Dickson Recurrence generates all factor coefficients from a single seed polynomial, combined with the Multiple Equal-Difference (MED) decomposition \cite{Zhu_MED} that partitions cyclotomic cosets for systematic trace dispatch. Second, the Ideal Derivation Modulo Principle (Theorem~\ref{thm:ideal_modulo}) provides the theoretical foundation linking the factorization coordinates to explicit polynomial ideals. Third, the Cofactor-Free Hensel Lift (Theorem~\ref{thm:cofactor_free}) enables independent per-factor precision elevation without Jacobian matrix inversion. Through inverse caching and fast modular exponentiation, the lifting cost per layer is reduced to $O(m^2)$---rendering it negligible relative to the base field seed extraction, and yielding an overall empirical base field complexity of $O(n + m^3 \log p)$ that is practically independent of the precision depth $e$.
\section{Preliminaries: Structural Polynomial Configurations}
We begin by defining the core polynomial recurrences that underpin the factorization framework.

\begin{defn}[Dickson Polynomials] \cite{Lidl1993} \label{def:dickson_poly}
Let $a \in \mathbb{F}_q$. The Dickson polynomials of the first kind $D_k(x, a)$ and the second kind $E_k(x, a)$ of degree $k$ are explicitly formulated through the symmetric parameterization of roots $x = y + a/y$:
\begin{align}
    D_k(x, a) &= y^k + \left(\frac{a}{y}\right)^k = \sum_{i=0}^{\lfloor k/2 \rfloor} \frac{k}{k-i} \binom{k-i}{i} (-a)^i x^{k-2i} \\
    E_k(x, a) &= \frac{y^{k+1} - (a/y)^{k+1}}{y - a/y} = \sum_{i=0}^{\lfloor k/2 \rfloor} \binom{k-i}{i} (-a)^i x^{k-2i}
\end{align}
Both $D_k$ and $E_k$ satisfy the common three-term recurrence $P_{k} = x\, P_{k-1} - a\, P_{k-2}$, differing only in their initial conditions: $(D_0, D_1) = (2, x)$ and $(E_0, E_1) = (1, x)$. Because this recurrence involves only multiplication and addition, it operates safely over integer rings $\mathbb{Z}_{p^e}$ without encountering zero-divisor issues.
\end{defn}

\begin{defn}[$q$-Cyclotomic Cosets and the Factorization of $X^n-1$] \cite{MacWilliams1977} \label{def:cyclotomic_cosets}
Let $n$ be a positive integer with $\gcd(n, q) = 1$. The \textbf{$q$-cyclotomic coset} of an element $s \in \mathbb{Z}_n$ is the set
\begin{equation}
    C_s = \{s, \, qs, \, q^2 s, \, \dots, \, q^{m_s - 1} s\} \pmod{n}
\end{equation}
where $m_s = |C_s|$ is the smallest positive integer satisfying $q^{m_s} s \equiv s \pmod{n}$. The collection of all distinct $q$-cyclotomic cosets partitions the index set $\mathbb{Z}_n = \{0, 1, \dots, n-1\}$.

The polynomial $X^n - 1$ decomposes as the product of all cyclotomic polynomials whose index divides $n$:
\begin{equation}
    X^n - 1 = \prod_{d \mid n} \Phi_d(X)
\end{equation}
where each $\Phi_d(X)$ is the $d$-th cyclotomic polynomial, the minimal polynomial of the primitive $d$-th roots of unity over $\mathbb{Q}$. Over a finite field $\mathbb{F}_q$, there is a natural bijection between the irreducible factors of $X^n - 1$ and the $q$-cyclotomic cosets of $\mathbb{Z}_n$: the coset $C_s$ of size $m_s$ corresponds to an irreducible factor of degree $m_s$ whose roots are $\{\zeta_n^{q^i s}\}_{i=0}^{m_s - 1}$, where $\zeta_n$ is a primitive $n$-th root of unity. Under this bijection, the irreducible factors of $\Phi_d(X)$ correspond precisely to those cosets $C_s$ for which $\zeta_n^s$ is a primitive $d$-th root of unity.
\end{defn}

\begin{defn}[Multivariate Dickson Recurrence] \cite{Macdonald1995} \label{def:dickson_multi}
Let $G(X) = X^m - A_1 X^{m-1} + \dots + (-1)^m A_m$ be a monic irreducible factor of $X^n-1$ over $\mathbb{F}_p$ (or $\mathbb{Z}_{p^e}$), with roots $\alpha, \alpha^p, \dots, \alpha^{p^{m-1}}$ forming a $q$-cyclotomic coset of size $m$ (Definition~\ref{def:cyclotomic_cosets}). The power sum traces $S_k = \sum_{i=0}^{m-1} \alpha^{p^i k}$ are determined in two phases.

\textbf{Initialization.} The first $m$ traces are computed from the seed coefficients $(A_1, \dots, A_m)$ via the Newton--Girard identities in their forward (multiplication-only) form:
\begin{equation} \label{eq:newton_girard_forward}
    S_k = \sum_{j=1}^{k-1} (-1)^{j-1} A_j\, S_{k-j} \;+\; (-1)^{k-1}\, k\, A_k, \qquad 1 \leq k \leq m
\end{equation}
where the base case $S_0 = \sum_{i=0}^{m-1} (\alpha^{p^i})^0 = m$ counts the roots. (Note: the forward recurrence~\eqref{eq:newton_girard_forward} does not reference $S_0$ directly, since the $j=k$ term is evaluated analytically as $(-1)^{k-1} k A_k$.) Since each $k A_k$ involves only multiplication by the integer $k$ (not division), this initialization is division-free and well-defined over $\mathbb{Z}_{p^e}$.

\textbf{LFSR propagation.} For all subsequent indices $k > m$, the traces satisfy the \textbf{division-free} linear recurrence:
\begin{equation} \label{eq:dickson_forward}
    S_k = A_1 S_{k-1} - A_2 S_{k-2} + \dots + (-1)^{m-1} A_m S_{k-m}
\end{equation}
This recurrence constitutes a Linear Feedback Shift Register (LFSR) whose characteristic polynomial is $G(X)$ itself. The name ``Multivariate Dickson Recurrence'' reflects the fact that this $m$-parameter recurrence generalizes the classical bivariate Dickson recurrence of Definition~\ref{def:dickson_poly}: when $m = 2$, setting $(A_1, A_2) = (x, a)$ recovers the two-term recurrence $P_k = x\, P_{k-1} - a\, P_{k-2}$. For general $m$, the roots specialize to conjugate tuples of roots of unity, and the recurrence coincides with the Generalized Chebyshev Polynomials studied in \cite{Chou1997, Fitzgerald2005}.

Crucially, this single recurrence encodes the coefficients of \emph{all} irreducible factors of the same dimension $m$, not merely those of $G(X)$ itself. For any other cyclotomic coset representative $s$, the corresponding factor $G_s(X)$ has roots $\alpha^s, \alpha^{ps}, \dots, \alpha^{p^{m-1}s}$, whose power sums satisfy $S_k^{(s)} = S_{sk \bmod n}$. Hence, the complete trace array $\{S_k\}_{k=1}^{n}$ generated from a single seed $G(X)$ contains the full coefficient information for every factor of dimension $m$ in the factorization of $X^n-1$.
\end{defn}

\begin{prop}[Dual-Track Coefficient Reconstruction] \label{prop:dual_track}
Let $G_s(X) = X^{m_s} - A_1^{(s)} X^{m_s - 1} + \dots + (-1)^{m_s} A_{m_s}^{(s)}$ be the irreducible factor of $X^n - 1$ over $\mathbb{Z}_{p^e}$ corresponding to a $q$-cyclotomic coset $C_s$ of size $m_s$, where the superscript $(s)$ indexes the coset representative~$s$ throughout. Let $S_1^{(s)}, \dots, S_{m_s}^{(s)} \in \mathbb{Z}_{p^e}$ be the dispatched power sum traces for coset $C_s$ (obtained via Proposition~\ref{prop:med_dispatch} below). Then the coefficients $(A_1^{(s)}, \dots, A_{m_s}^{(s)})$ can be uniquely recovered as follows.

\begin{enumerate}
    \item[\textup{(1)}] \textbf{Primary Track} ($p > m_s$). The Newton--Girard identities \cite{Macdonald1995} provide an explicit inversion:
    \begin{equation} \label{eq:newton_girard}
        A_k^{(s)} = \frac{1}{k} \sum_{j=1}^{k} (-1)^{j-1} A_{k-j}^{(s)}\, S_j^{(s)}, \qquad k = 1, \dots, m_s
    \end{equation}
    Since $p > m_s$ implies $\gcd(k, p) = 1$ for all $k \in \{1, \dots, m_s\}$, each $1/k$ is well-defined in $\mathbb{Z}_{p^e}$. The reconstruction cost is $O(m_s)$ per factor.
    
    \item[\textup{(2)}] \textbf{Fallback Track} ($p \leq m_s$). Newton--Girard inversion fails because at least one $k \in \{1, \dots, m_s\}$ satisfies $p \mid k$, rendering $1/k$ non-invertible. Instead, represent the root $\alpha^s$ as $X^s \bmod G_{\text{seed}}(X)$ in the quotient ring $\mathbb{Z}_{p^e}[X]/\langle G_{\text{seed}}(X) \rangle$ and compute successive powers to form the $(m_s + 1) \times m$ coefficient matrix $\mathbf{M}$. Then $G_s(X)$ is the unique monic polynomial of degree $m_s$ in the left null space of $\mathbf{M}$, recovered via Gaussian elimination. The Frobenius endomorphism guarantees $\det(\mathbf{M}) \not\equiv 0 \pmod{p}$, ensuring unconditional solvability. The reconstruction cost is $O(m_s^2 \cdot m)$ per factor.
\end{enumerate}
Together, these two tracks ensure that coefficient reconstruction succeeds for all characteristic regimes without exception.
\end{prop}

\section{Mathematical Baseline and Complexity Bounds}

\subsection{Discrete Field Benchmarks: $\mathbb{F}_p[X]$}
The factorization of a degree-$n$ polynomial over a finite field $\mathbb{F}_q$ ($q = p^k$) is bounded by two classical paradigms. Berlekamp's deterministic algorithm \cite{Berlekamp1968} constructs the $n \times n$ Berlekamp matrix $Q$ whose null space yields the factorization, at a cost of $O(n^3 + n^2 \log q)$ field operations (the $n^3$ term arising from Gaussian elimination, the $n^2 \log q$ term from computing $X^{iq} \bmod f(X)$ via repeated squaring). The Cantor-Zassenhaus probabilistic algorithm \cite{Cantor1981} replaces the exhaustive null-space enumeration with random splitting via $\gcd(f, g^{(q-1)/2} - 1)$, achieving an expected complexity of $O(n^2 \log q)$ with fast polynomial arithmetic \cite{Gathen1998}.

In contrast to these global algorithms, our framework separates base field factorization into two distinct phases. The first phase, seed discovery (Stage 1), relies on randomized primitivity testing within the extension field $\mathbb{F}_{p^m}$, where exponentiation costs $O(m \log p)$ multiplications. Because each multiplication in $\mathbb{F}_{p^m}$ requires $O(m^2)$ base field operations, the total algebraic cost for seed discovery is bounded by $O(m^3 \log p)$. The second phase, coefficient expansion, generates all remaining factors from the seed via the Dickson LFSR and dual-track reconstruction (Stages 3 and 4). Trace generation requires $O(n)$ base field operations, and recovering the $n/m$ factors demands $O(m^2)$ operations per factor, contributing $O(n \cdot m)$. For typically small or constant coset dimensions $m$, this expansion phase strictly scales as $O(n)$. Consequently, the composite base field algebraic complexity of our framework is bounded by $O(n + m^3 \log p)$.

\subsection{Lifting Complexity Bounds over $\mathbb{Z}_{p^e}[X]$}
Lifting factorizations from $\mathbb{F}_p$ to $\mathbb{Z}_{p^e}$ faces fundamental complexity constraints that depend on the chosen methodology.

\textbf{Classical polynomial Hensel lifting} \cite{Cohen1993} requires maintaining the complementary cofactor $H(X) = (X^n-1)/G(X)$ and the B\'{e}zout coefficients $u, v$ satisfying $u G + v H \equiv 1$. Updating these auxiliary polynomials at each precision layer adds a cost heavily dependent on the global degree $n$. Alternatively, \textbf{multivariate coefficient Hensel lifting} attempts to bypass the cofactor by solving directly for the factor's $m$ unknown coefficients. However, when the coset dimension satisfies $m \ge 2$, this requires computing and inverting an $m \times m$ Jacobian matrix of partial derivatives at each layer. Over $\mathbb{Z}_{p^e}$, where zero-divisors obstruct standard matrix inversion, resolving the singularities demands $O(p^{(m-1)^2})$ operations.

\textbf{McGuire's single-factor approach} \cite{McGuire2001} eliminates the cofactor dependency by exploiting the Graeffe root-squaring method: given a factor $G_k(X)$ accurate to precision $p^k$, the factor at precision $p^{k+1}$ is obtained by computing $\prod_{j=1}^{p} G_k(\omega^j X)$ (where $\omega$ is a primitive $p$-th root of unity) and reducing the result. This achieves cofactor-free single-factor lifting, but requires $O(p)$ polynomial multiplications per layer, making the per-layer cost $O(p \cdot m \log(pm))$. For the small primes typical of coding theory ($p = 2, 3$), this is efficient; however, for the large primes demanded by modern cryptography ($p \gg m$), the linear dependence on $p$ renders the method impractical.

In contrast, the cofactor-free Hensel lift (Theorem~\ref{thm:cofactor_free}) developed in this framework elevates a single seed factor independently of the global degree $n$. By utilizing a cached polynomial inverse computed once over $\mathbb{F}_p$, each precision layer requires only basic polynomial arithmetic modulo the seed factor, taking $O(m^2)$ algebraic operations. Across $e$ precision layers, the total lifting cost is bounded by $O(e \cdot m^2)$. 

Combining this with the base field complexity derived in Section 2.1, the grand total algebraic complexity for explicitly factoring $X^n-1$ over $\mathbb{Z}_{p^e}$ is strictly bounded by $O(n + m^3 \log p + e \cdot m^2)$. This composite bound clearly isolates the independent contributions of the global degree ($n$), field characteristic ($p$), and precision depth ($e$), formally guaranteeing the framework's immunity against the exponential explosions characteristic of multivariate Jacobian inversion.

\subsection{Cyclotomic Specializations}
For the factorization of $X^n-1$, classical approaches exploit the Number Theoretic Transform (NTT) to evaluate polynomials at roots of unity. However, NTT-based methods decompose $X^n-1$ by evaluating at individual roots $\omega^k \in \mathbb{F}_{p^m}$, requiring explicit construction of the extension field when $m \ge 2$. For coset dimensions $m \ge 3$, the extension field arithmetic and the bookkeeping of conjugate roots become the dominant cost, negating the NTT's efficiency advantages. Our framework sidesteps extension field construction entirely by operating through the LFSR-based Dickson recurrence (Definition~\ref{def:dickson_multi}), which generates all required traces directly from the seed polynomial coefficients over $\mathbb{F}_p$ (or $\mathbb{Z}_{p^e}$).

\section{Factorization over Arbitrary $n$: Cyclotomic Coset Dispatch}

Extending the factorization of $X^n-1$ to arbitrary $n$ requires understanding how the $q$-cyclotomic cosets partition the index set $\{0, 1, \dots, n-1\}$ and how the global trace array is dispatched to reconstruct individual factors.

\subsection{Cyclotomic Coset Decomposition}
For the unramified case $\gcd(n, p) = 1$, each index $s \in \{0, \dots, n-1\}$ belongs to exactly one $q$-cyclotomic coset $C_s = \{s, ps, \dots, p^{m_s-1}s \} \pmod{n}$, where $m_s = |C_s|$ denotes the coset size. Each coset $C_s$ corresponds to an irreducible factor of degree $m_s$. In general, different cosets may have different sizes, producing a \emph{mixed-degree} decomposition---this occurs not only for general $n$ but even in the structured case $n=p+1$, where the cosets $C_0 = \{0\}$ and $C_{n/2} = \{n/2\}$ degenerate to singletons corresponding to the linear factors $X-1$ and $X+1$.

\subsection{Trace Dispatch via the MED Structure}

\begin{prop}[MED Trace Dispatch] \label{prop:med_dispatch}
Let $G_{\text{seed}}(X)$ be the lifted seed of degree $m$, and let $\{S_k\}_{k=1}^{n}$ be the global trace array generated via Definition~\ref{def:dickson_multi}. For a $q$-cyclotomic coset $C_s$ of size $m_s$, define the intrinsic trace $T_k^{(s)} = \sum_{i=0}^{m_s-1} (\alpha^s)^{p^i k}$. Then:

\begin{enumerate}
    \item[\textup{(1)}] \textbf{Full-rank dispatch} ($m_s = m$). The intrinsic traces are obtained by direct index scaling:
    \begin{equation} \label{eq:trace_dispatch}
        T_k^{(s)} = S_{sk \bmod n}, \qquad k = 1, \dots, m
    \end{equation}
    
    \item[\textup{(2)}] \textbf{Degenerate dispatch} ($m_s < m$). The root $\alpha^s$ satisfies $(\alpha^s)^{p^{m_s}} = \alpha^s$, so the $m$-dimensional trace aggregates each root $m/m_s$ times:
    \begin{equation} \label{eq:degenerate_scaling}
        S_{sk \bmod n} = \frac{m}{m_s} \cdot T_k^{(s)}, \qquad k = 1, \dots, m_s
    \end{equation}
    The scaling factor $m/m_s$ is invertible over $\mathbb{Z}_{p^e}$ since $\gcd(m/m_s, \, p) = 1$ in the unramified setting \cite{Zhu_MED}. Consequently, the intrinsic traces are recovered as $T_k^{(s)} = (m/m_s)^{-1} \cdot S_{sk \bmod n}$.
\end{enumerate}
In both cases, the dispatched traces $T_1^{(s)}, \dots, T_{m_s}^{(s)}$ are then fed into the dual-track coefficient reconstruction (Proposition~\ref{prop:dual_track}) to recover the factor $G_s(X)$.
\end{prop}

\begin{example}[Degenerate Coset Dispatch]
Consider $p=5, n=12$. The seed has dimension $m=2$ (since $5^2 \equiv 1 \pmod{12}$). The coset $C_3 = \{3\}$ has size $m_s = 1$ (since $5 \times 3 = 15 \equiv 3 \pmod{12}$), corresponding to a linear factor $X - c$. The global trace $S_3 = \alpha^3 + (\alpha^3)^5 = 2\alpha^3$, which aggregates the single root $\alpha^3$ twice ($m/m_s = 2$). Dividing by $2$ (i.e., multiplying by $2^{-1} \equiv 3 \pmod{5}$) recovers $\alpha^3 = 3S_3 \pmod{5}$, yielding the linear factor $(X - \alpha^3)$.
\end{example}

\section{Structural Analysis: Multivariate Ideals and Their Jacobian Barrier}

In our foundational work \cite{Wang2026ISIT}, the factorization of $X^{p+1}-1$ over $\mathbb{Z}_{p^e}$ was driven by a single auxiliary polynomial $V(x)$, whose roots over $\mathbb{F}_p$ directly enumerated all valid structural values of the quadratic factors, enabling Jacobian-free Hensel lifting. However, this elegant one-dimensional mechanism was specific to $n=p+1$ and $m=2$. We now ask: \emph{where does $V(x)$ come from algebraically, and how does it generalize to arbitrary $n$ and $m$?}

The answer is the Ideal Derivation Modulo Principle, which derives the complete polynomial ideal governing factorization coordinates for any $(n, m)$. We first state and prove the general theorem, then demonstrate that the V1 structural polynomial $V(x)$ emerges as its exact 1-dimensional specialization. Finally, we show that solving this general ideal via Jacobian methods leads to an exponential barrier, motivating the local cofactor-free approach of Section~\ref{sec:local_lift}.

\subsection{Theorem of Ideal Derivation Modulo Principle}

\begin{thm}[Ideal Derivation Modulo Principle] \label{thm:ideal_modulo}
Let $G(\mathbf{A}, X) = X^m - A_1 X^{m-1} + \dots + (-1)^m A_m$ be a monic degree-$m$ polynomial whose coefficients $\mathbf{A} = (A_1, \dots, A_m)$ are treated as indeterminates. Performing polynomial long division of $X^n - 1$ by $G(\mathbf{A}, X)$ in $X$ yields a remainder of degree at most $m-1$:
\begin{equation}
    R(\mathbf{A}, X) = V_{m-1}(\mathbf{A})\, X^{m-1} + \dots + V_1(\mathbf{A})\, X + V_0(\mathbf{A})
\end{equation}
where each $V_i(\mathbf{A})$ is an explicit polynomial in the ring $\mathbb{F}_p[A_1, \dots, A_m]$. These $m$ polynomials generate the ideal
\begin{equation}
    I_{n,m} = \langle V_0, \, V_1, \, \dots, \, V_{m-1} \rangle \;\subseteq\; \mathbb{F}_p[A_1, \dots, A_m]
\end{equation}
Then $G(\mathbf{A}^*, X)$ is a monic factor of $X^n - 1$ over $\mathbb{F}_p$ if and only if $\mathbf{A}^*$ lies in the algebraic variety of $I_{n,m}$, i.e.,
\begin{equation}
    V_0(\mathbf{A}^*) = V_1(\mathbf{A}^*) = \dots = V_{m-1}(\mathbf{A}^*) = 0 \quad \text{in } \mathbb{F}_p
\end{equation}
Consequently, every degree-$m$ irreducible factor of $X^n-1$ over $\mathbb{F}_p$ is uniquely identified by a common root $\mathbf{A}^* \in \mathbb{F}_p^m$ of this ideal, and each such root produces exactly one factor.
\end{thm}

\begin{proof}
Let $G(\mathbf{A}, X) = X^m - A_1 X^{m-1} + \dots + (-1)^m A_m$ be a monic polynomial of degree $m$ whose coefficients are parameterized by the elementary symmetric polynomials $\mathbf{A} = (A_1, \dots, A_m)$ of its roots. Performing polynomial long division of $X^n - 1$ by $G(\mathbf{A}, X)$ in the indeterminate $X$ yields:
\begin{equation}
    X^n - 1 = Q(\mathbf{A}, X) \cdot G(\mathbf{A}, X) + R(\mathbf{A}, X)
\end{equation}
where $Q(\mathbf{A}, X)$ is the quotient and $R(\mathbf{A}, X)$ is the remainder. Since $G$ is monic of degree $m$, the division algorithm guarantees $\deg_X R < m$, so:
\begin{equation}
    R(\mathbf{A}, X) = V_{m-1}(\mathbf{A})\, X^{m-1} + V_{m-2}(\mathbf{A})\, X^{m-2} + \dots + V_0(\mathbf{A})
\end{equation}
where each coefficient $V_i(\mathbf{A})$ is a polynomial expression in $A_1, \dots, A_m$. Concretely, the division proceeds by the iterative reduction $X^k \equiv A_1 X^{k-1} - A_2 X^{k-2} + \dots + (-1)^{m-1} A_m X^{k-m}$ for $k \geq m$, applied successively to $X^n, X^{n-1}, \dots$ until all terms are reduced below degree $m$.

Now suppose that a specific assignment $\mathbf{A} = \mathbf{A}^*$ produces an irreducible factor $G(\mathbf{A}^*, X)$ that divides $X^n - 1$ over $\mathbb{F}_p$. Then $X^n - 1 \equiv 0 \pmod{G(\mathbf{A}^*, X)}$, which forces:
\begin{equation}
    R(\mathbf{A}^*, X) = V_{m-1}(\mathbf{A}^*)\, X^{m-1} + \dots + V_0(\mathbf{A}^*) \equiv 0
\end{equation}
as a polynomial identity in $X$. Since $\{1, X, X^2, \dots, X^{m-1}\}$ are linearly independent over $\mathbb{F}_p$ (their degrees are pairwise distinct), the vanishing of $R$ requires each coefficient to vanish independently:
\begin{equation}
    V_i(\mathbf{A}^*) = 0, \qquad i = 0, 1, \dots, m-1
\end{equation}
Therefore, the admissible coordinate values $\mathbf{A}^*$ are precisely the common zeros of the polynomial system $\{V_0 = 0, \, V_1 = 0, \, \dots, \, V_{m-1} = 0\}$. This system constitutes the polynomial ideal governing the factorization coordinates over $\mathbb{F}_p$.
\end{proof}

\subsection{Exact Equivalence with the V1 Structural Polynomial $V(x)$}

We now demonstrate that the V1 structural polynomial $V(x)$ from \cite{Wang2026ISIT} is the exact 1-dimensional specialization of the Ideal Derivation Modulo Principle. The derivation proceeds in five precise steps.

\textbf{Step 1: Restrict to $m=2$.} For a quadratic factor $G(S, P, X) = X^2 - SX + P$ parameterized by the trace $S = A_1$ and the norm $P = A_2$, the division of $X^n-1$ by $G$ yields a degree-1 remainder:
\begin{equation}
    X^n - 1 \equiv V_1(S, P)\, X + V_0(S, P) \pmod{X^2 - SX + P}
\end{equation}
The Ideal Derivation Modulo Principle produces two constraints: $V_1(S, P) = 0$ and $V_0(S, P) = 0$.

\textbf{Step 2: Evaluate $V_1$ and $V_0$ explicitly.} By recursively applying the reduction $X^k \equiv S X^{k-1} - P X^{k-2}$, each power $X^k$ reduces to a linear form $X^k \equiv a_k X + b_k \pmod{X^2 - SX + P}$, where the coefficients $a_k$ and $b_k$ satisfy the same recurrence $a_k = S a_{k-1} - P a_{k-2}$ with initial conditions $a_0 = 0, a_1 = 1$ (and $b_0 = 1, b_1 = 0$ respectively). Recognizing that the recurrence for $a_k$ is precisely the defining recurrence of the Dickson polynomial of the second kind $E_k(S, P)$ (Definition~\ref{def:dickson_poly}), we identify $a_k = E_{k-1}(S, P)$. The detailed computation is given in Appendix~\ref{app:dickson_derivation}. The result is:
\begin{align}
    V_1(S, P) &= E_{n-1}(S, P) \\
    V_0(S, P) &= -P \cdot E_{n-2}(S, P) - 1
\end{align}

\textbf{Step 3: Impose the cyclotomic norm constraint $P=1$.} For $n = p+1$, each non-trivial cyclotomic coset has the form $C_s = \{s, \, ps \equiv -s \pmod{p+1}\}$. The corresponding quadratic factor has roots $\zeta^s$ and $\zeta^{-s}$, whose product is $\zeta^s \cdot \zeta^{-s} = 1$. Therefore, the norm is structurally fixed: $P = 1$ for every quadratic factor. After substituting $P = 1$, the two-variable system collapses to a single-variable system:
\begin{align}
    V_1(S, 1) &= E_{n-1}(S, 1) = 0 \\
    V_0(S, 1) &= -E_{n-2}(S, 1) - 1 = 0
\end{align}

\textbf{Step 4: Bridge to the structural variable.} The constraint $E_p(A_i, 1) = 0$ acts on the trace $A_i$. In \cite{Wang2026ISIT}, however, the structural polynomial $V(x)$ acts on the \emph{structural variable} $S_i = 2 - A_i^2$. These two formulations are linked by the Dickson identity $D_2(A_i, 1) = A_i^2 - 2$, which gives $S_i = -D_2(A_i, 1) = -A_{2i}$.

The connection arises from the conjugate pairing structure of $n=p+1$. Since $\alpha^{(p+1)/2} = -1$, the index pair $(i, \frac{p+1}{2}-i)$ satisfies $A_j = -A_i$, so their quadratic factors multiply as:
\begin{equation}
    (X^2 - A_i X + 1)(X^2 + A_i X + 1) = X^4 + S_i X^2 + 1
\end{equation}
This pairing reduces the $\frac{p-1}{2}$ independent trace values to $k = \lfloor \frac{p}{4} \rfloor$ independent structural variables $S_1, \dots, S_k$. The product of all paired quartic factors yields:
\begin{equation}
    \prod_{i=1}^{k} (X^4 + S_i X^2 + 1) \equiv \sum_{j=0}^{k} X^{4j} \pmod{p}
\end{equation}
Expanding the left side and matching coefficients shows that the $S_i$ are precisely the roots of a degree-$k$ polynomial whose coefficients are the elementary symmetric polynomials determined by alternating binomial identities---this is exactly the structural polynomial $V(x)$ of \cite{Wang2026ISIT} (Definition~2 therein). Thus, the V1 polynomial $V(x)$ is obtained from our $E_p(A_i, 1) = 0$ by the substitution $x = 2 - A_i^2$ and the subsequent conjugate pairing reduction. The case $A_i = 0$ (yielding $S_i = 2$ and $X^2 + 1$) arises precisely when $4 \mid (p+1)$.

\textbf{Step 5: Redundancy of $V_0$.} The constraint $V_0(S, 1) = 0$ is automatically satisfied whenever $V_1(S, 1) = 0$ and $P = 1$. This follows because if $E_{p}(S, 1) = 0$, then $G(X) = X^2 - SX + 1$ divides $X^{p+1} - 1$, which forces the \emph{entire} remainder to vanish, including its constant term $V_0$. Consequently, the two-constraint system $\{V_1 = 0, V_0 = 0\}$ reduces to the single scalar equation $E_p(S, 1) = 0$ over $\mathbb{F}_p$.

This completes the exact equivalence chain: the Ideal Derivation Modulo Principle (Theorem~\ref{thm:ideal_modulo}) produces $E_p(A_i, 1) = 0$ on the trace variable $A_i$; the Dickson substitution $S_i = 2 - A_i^2$ and conjugate pairing reduction then collapse this into the structural polynomial $V(S_i) = 0$ of \cite{Wang2026ISIT}. The general theorem extends this mechanism to arbitrary $(n, m)$ by providing $m$ independent constraints on $m$ free coordinates---but at the cost of requiring multivariate Jacobian inversion to solve the resulting system.

\subsection{The Jacobian Complexity Barrier}

Theorem~\ref{thm:ideal_modulo} provides a mathematically complete characterization of the factorization coordinates. However, \emph{solving} the polynomial system $\{V_i(\mathbf{A}) = 0\}$ over $\mathbb{Z}_{p^e}$ via multivariate Newton--Hensel iteration requires evaluating and inverting the $m \times m$ Jacobian matrix $J = (\partial V_i / \partial A_j)$ at each lifting layer. Over integer rings permeated by zero-divisors, this inversion demands exhaustive enumeration bounded by $O(p^{(m-1)^2})$ per layer---a cost that grows exponentially in the coset dimension $m$ and becomes computationally prohibitive even for moderate $m \geq 3$.

This exponential barrier is not an artifact of a particular algorithm but a structural consequence of the global multivariate approach: any method that attempts to lift \emph{all} symmetric coordinates $\mathbf{A}$ simultaneously must contend with the coupled nonlinear structure of the ideal. The framework developed in the following section circumvents this barrier entirely by abandoning the global ideal in favor of local, per-factor Hensel lifting.

\begin{rem}[Scalar collapse for $n=p+1$]
For the specific boundary $n=p+1$ with $m=2$, the cyclotomic constraint $P=1$ (Section~5.2, Step~3) eliminates one coordinate, reducing the $2 \times 2$ Jacobian to the scalar derivative $V_1'(S) = E_p'(S, 1)$. This recovers the V1 mechanism of \cite{Wang2026ISIT}: the Hensel update $S^{(h+1)} = S^{(h)} - V_1(S^{(h)}) / V_1'(S^{(1)})$ requires only a single pre-computed inverse, bypassing matrix inversion entirely. Notably, the cofactor-free lift of Section~\ref{sec:local_lift} achieves an analogous single-inverse property for arbitrary $m$: its cached polynomial inverse $[H(X)]^{-1} \bmod G_1(X)$ is, by the product rule identity $f' \equiv G_1' H \pmod{G_1}$, equivalent to $G_1'(X) \cdot [f'(X)]^{-1} \bmod G_1(X)$---the polynomial-ring counterpart of the scalar derivative inverse.
\end{rem}

\begin{example}[Higher-Dimensional Extraction: $p=2, n=7, m=3$]
For $p=2, n=7$, the multiplicative order of $2$ modulo $7$ is $3$, yielding coset dimension $m=3$---beyond the reach of V1's scalar $V(x)$ mechanism. The Ideal Derivation Modulo Principle produces three constraints $V_0(A_1, A_2, A_3) = V_1(A_1, A_2, A_3) = V_2(A_1, A_2, A_3) = 0$. Over $\mathbb{F}_2$, the framework extracts:
\begin{equation}
    X^7-1 \equiv (X+1)(X^3+X+1)(X^3+X^2+1) \pmod 2
\end{equation}
\end{example}

\begin{example}[Mixed-Degree Decomposition: $p=2, n=21$]
For $p=2, n=21$, the cyclotomic cosets have four distinct sizes ($m_s \in \{1, 2, 3, 6\}$), producing a mixed-degree factorization:
\begin{align}
    X^{21}-1 \equiv \ &(X+1)(X^2+X+1)(X^3+X+1)(X^3+X^2+1) \notag \\
                      &(X^6+X^4+X^2+X+1)(X^6+X^5+X^4+X^2+1) \pmod 2
\end{align}
The coset dispatch mechanism (Section~4) handles each dimension independently: the degree-1 and degree-2 factors are recovered via MED scaling from the degree-3 seed's trace array, while the degree-6 factors require a separate seed of matching dimension.
\end{example}

\section{Cofactor-Free Hensel Lift} \label{sec:local_lift}

Rather than solving the global multivariate ideal of Section~5, we adopt a fundamentally different strategy: each irreducible factor is lifted independently from $\mathbb{F}_p$ to $\mathbb{Z}_{p^e}$ via a local, cofactor-free Hensel mechanism. Figure~\ref{fig:lift_compare} illustrates the structural contrast between the two approaches.

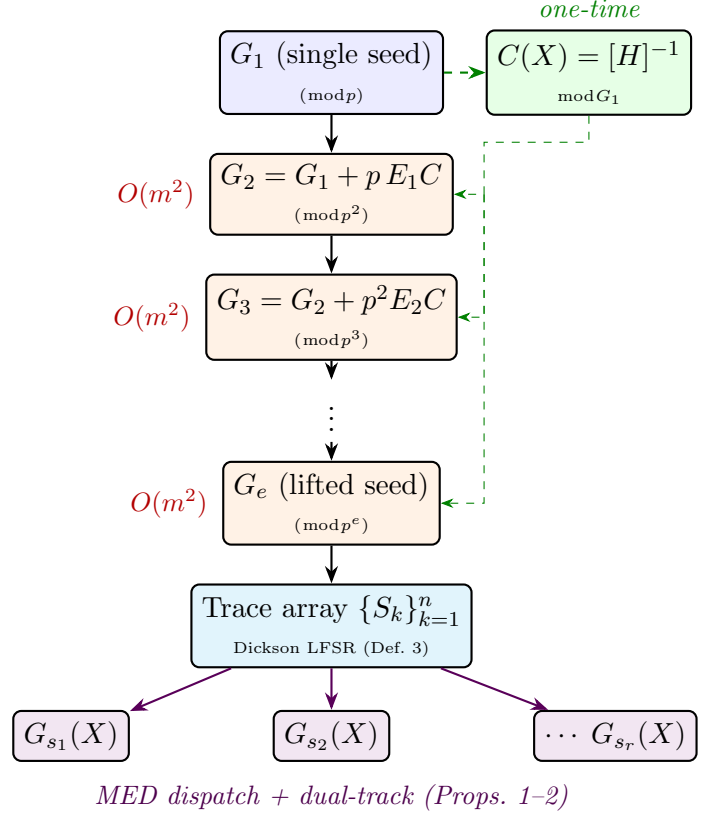
\begin{figure}[htbp]
\centering
\begin{tikzpicture}[%
  >=Stealth,
  node distance=0.5cm and 1.6cm,
  every node/.style={font=\small},
  box/.style={draw, rounded corners=3pt, minimum width=2.6cm, minimum height=0.6cm, align=center, thick},
  startbox/.style={box, fill=blue!8},
  layerbox/.style={box, fill=orange!10},
  resultbox/.style={box, fill=violet!10, minimum width=1.5cm},
  tracebox/.style={box, fill=cyan!10},
  costlabel/.style={font=\footnotesize\itshape, text=red!70!black},
  cachelabel/.style={font=\footnotesize\itshape, text=green!50!black},
  titlebox/.style={font=\bfseries\small, align=center},
]

\node[titlebox] (titleL) {Classical Hensel Lifting};
\node[startbox, below=0.35cm of titleL] (startL) {$G_1^{(i)},\; H_1^{(i)},\; u,v$\\{\tiny($\bmod p$)}};
\node[layerbox, below=of startL] (L1) {Update $G_2^{(i)}, H_2^{(i)}, u, v$\\{\tiny($\bmod p^2$)}};
\node[layerbox, below=of L1] (L2) {Update $G_3^{(i)}, H_3^{(i)}, u, v$\\{\tiny($\bmod p^3$)}};
\node[below=0.25cm of L2] (dotsL) {$\vdots$};
\node[layerbox, below=0.25cm of dotsL] (Le) {Update $G_e^{(i)}, H_e^{(i)}, u, v$\\{\tiny($\bmod p^e$)}};

\draw[->, thick] (startL) -- (L1);
\draw[->, thick] (L1) -- (L2);
\draw[->, thick] (L2) -- (dotsL);
\draw[->, thick] (dotsL) -- (Le);

\node[costlabel, right=0.05cm of L1] {$O(n)$};
\node[costlabel, right=0.05cm of L2] {$O(n)$};
\node[costlabel, right=0.05cm of Le] {$O(n)$};

\node[below=0.35cm of Le, font=\footnotesize\bfseries, text=red!60!black, align=center] (repeatL) {Repeat for each factor\\$i = 1, 2, \dots, r$};
\draw[->, thick, red!60!black] (Le) -- (repeatL);

\node[titlebox, right=5.0cm of titleL] (titleR) {Cofactor-Free Lifting (Ours)};
\node[startbox, below=0.35cm of titleR] (startR) {$G_1$ (single seed)\\{\tiny($\bmod p$)}};

\node[box, fill=green!10, right=0.55cm of startR, minimum width=2.0cm] (cache) {$C(X) = [H]^{-1}$\\{\tiny $\bmod G_1$}};
\draw[->, thick, dashed, green!50!black] (startR.east) -- (cache.west);
\node[cachelabel, above=0.02cm of cache] {one-time};

\node[layerbox, below=of startR] (R1) {$G_2 = G_1 + p\, E_1 C$\\{\tiny($\bmod p^2$)}};
\node[layerbox, below=of R1] (R2) {$G_3 = G_2 + p^2 E_2 C$\\{\tiny($\bmod p^3$)}};
\node[below=0.25cm of R2] (dotsR) {$\vdots$};
\node[layerbox, below=0.25cm of dotsR] (Re) {$G_e$ (lifted seed)\\{\tiny($\bmod p^e$)}};

\draw[->, thick] (startR) -- (R1);
\draw[->, thick] (R1) -- (R2);
\draw[->, thick] (R2) -- (dotsR);
\draw[->, thick] (dotsR) -- (Re);

\draw[->, dashed, green!50!black] (cache.south) -- ++(0,-0.35) -| ([xshift=0.4cm]R1.east) -- (R1.east);
\draw[->, dashed, green!50!black] ([xshift=0.4cm]R1.east) |- (R2.east);
\draw[->, dashed, green!50!black] ([xshift=0.4cm]R1.east) |- (Re.east);

\node[costlabel, left=0.05cm of R1] {$O(m^2)$};
\node[costlabel, left=0.05cm of R2] {$O(m^2)$};
\node[costlabel, left=0.05cm of Re] {$O(m^2)$};

\node[tracebox, below=0.5cm of Re] (trace) {Trace array $\{S_k\}_{k=1}^{n}$\\{\tiny Dickson LFSR (Def.~3)}};
\draw[->, thick] (Re) -- (trace);

\node[resultbox, below left=0.55cm and 0.8cm of trace] (f1) {$G_{s_1}(X)$};
\node[resultbox, below=0.55cm of trace] (f2) {$G_{s_2}(X)$};
\node[resultbox, below right=0.55cm and 0.8cm of trace] (f3) {$\cdots\; G_{s_r}(X)$};

\draw[->, thick, violet!70!black] (trace) -- (f1);
\draw[->, thick, violet!70!black] (trace) -- (f2);
\draw[->, thick, violet!70!black] (trace) -- (f3);

\node[font=\footnotesize\itshape, text=violet!60!black, below=0.15cm of f2] {MED dispatch + dual-track (Props.~1--2)};

\end{tikzpicture}
\caption{Structural comparison of lifting strategies. \textbf{Left:} Classical polynomial Hensel lifting maintains and updates the factor $G^{(i)}$, cofactor $H^{(i)}$, and B\'{e}zout coefficients $u, v$ at every precision layer, incurring $O(n)$ cost per layer; this procedure must be repeated independently for each of the $r$ irreducible factors. \textbf{Right:} Our cofactor-free approach lifts a \emph{single} seed factor using a cached inverse $C(X)$ precomputed once over $\mathbb{F}_p$ (dashed green arrows). At the target precision $p^e$, the lifted seed generates the global trace array $\{S_k\}$ via the Dickson LFSR, from which \emph{all} remaining factors are reconstructed simultaneously via MED dispatch and dual-track inversion (purple arrows), without any additional Hensel lifts.}
\label{fig:lift_compare}
\end{figure}

The complete four-stage pipeline is formalized in Algorithm~\ref{alg:main}, whose stages correspond directly to the flow depicted in Figure~\ref{fig:lift_compare} (right).

\begin{algorithm}[h!]
\caption{Cofactor-Free Factorization of $X^n-1$ over $\mathbb{Z}_{p^e}$} \label{alg:main}
\KwIn{Cyclotomic order $n$, prime $p$ with $\gcd(n, p) = 1$, precision depth $e$}
\KwOut{Complete irreducible factorization of $X^n-1$ over $\mathbb{Z}_{p^e}$}
\tcp{Stage 1: Seed Discovery over $\mathbb{F}_p$}
Find a single irreducible factor $G_1(X) \in \mathbb{F}_p[X]$ of $X^n-1$ of degree $m$ via randomized primitivity testing ($X^n \equiv 1, \; X^{n/q} \not\equiv 1 \bmod G_1$ for all prime divisors $q \mid n$)\;
\tcp{Stage 2: Single-Seed Hensel Lift (Theorem~\ref{thm:cofactor_free})}
$C(X) \leftarrow [H(X)]^{-1} \bmod G_1(X)$, where $H(X) = (X^n-1)/G_1(X) \bmod p$ \quad \tcp{Cached once}
$G_{lift}(X) \leftarrow G_1(X)$\;
\For{$h = 1$ \KwTo $e-1$}{
  $E_h(X) \leftarrow (X^n \bmod G_{lift}(X) - 1) / p^h \bmod p$\;
  $\Delta G_h(X) \leftarrow E_h(X) \cdot C(X) \bmod G_1(X)$ \quad \tcp{Computed over $\mathbb{F}_p$}
  $G_{lift}(X) \leftarrow G_{lift}(X) + p^h \cdot \Delta G_h(X) \bmod p^{h+1}$\;
}
\tcp{Stage 3: Trace Generation from Lifted Seed}
Extract elementary symmetric polynomials $\{A_k\}$ from $G_{lift}(X)$ over $\mathbb{Z}_{p^e}$\;
Generate global trace array $\{S_k\}_{k=1}^{n}$ via the initialization \eqref{eq:newton_girard_forward} and LFSR propagation \eqref{eq:dickson_forward} of Definition~\ref{def:dickson_multi}\;
\tcp{Stage 4: Dual-Track Factor Reconstruction}
\ForEach{cyclotomic coset $C_s$ with dimension $m_s$}{
  Extract sub-traces $\{S_{s \cdot k \bmod n}\}$ via MED scaling (Proposition~\ref{prop:med_dispatch})\;
  \eIf{$p > m_s$}{
    Recover factor coefficients via Newton-Girard inversion\;
  }{
    Recover via quotient ring minimal polynomial (Gaussian elimination in $\mathbb{Z}_{p^e}[X]/\langle G_{lift} \rangle$)\;
  }
}
\Return{Complete factorization of $X^n-1$ over $\mathbb{Z}_{p^e}$}
\end{algorithm}

As illustrated by the fan-out in Figure~\ref{fig:lift_compare}, the lifted seed $G_{lift}(X)$ produced by Stage~2 of Algorithm~\ref{alg:main} serves as the sole input from which all remaining factors are generated---via Dickson trace recurrence (Stage~3) and dual-track coefficient reconstruction (Stage~4)---without additional Hensel lifts.

\subsection{Formalization: Cofactor-Free Lifting}

Traditional Hensel lifting requires maintaining and updating both the factor $G(X)$ and its complementary cofactor $H(X) = (X^n-1)/G(X)$ at each precision layer, with per-layer cost scaling with $\deg(H) = n - m$. The following theorem shows that the cofactor need only be accessed once---to compute a cached inverse over $\mathbb{F}_p$---after which all subsequent lifting iterations proceed without any cofactor involvement.

\begin{thm}[Cofactor-Free Hensel Lift] \label{thm:cofactor_free}
Let $G_1(X) \in \mathbb{F}_p[X]$ be a monic irreducible factor of $X^n - 1$ of degree $m$, and let $H(X) = (X^n - 1)/G_1(X) \in \mathbb{F}_p[X]$ be the complementary cofactor. Define the cached inverse
\begin{equation}
    C(X) \;=\; \bigl[H(X)\bigr]^{-1} \bmod G_1(X) \quad \text{in } \mathbb{F}_p[X]/\langle G_1(X) \rangle
\end{equation}
which exists since $\gcd(G_1, H) = 1$ in $\mathbb{F}_p[X]$. Suppose $G_h(X) \in \mathbb{Z}_{p^h}[X]$ is a monic lift satisfying $X^n - 1 \equiv 0 \pmod{G_h(X), \, p^h}$. Then:

\begin{enumerate}
    \item[\textup{(1)}] The error polynomial at layer $h$,
    \begin{equation}
        E_h(X) \;\equiv\; \frac{X^n \bmod G_h(X) \;-\; 1}{p^h} \pmod{p}
    \end{equation}
    depends on the current approximation $G_h(X)$ and thus \emph{changes at every layer}.
    
    \item[\textup{(2)}] The next-layer lift is obtained by computing the correction in $\mathbb{F}_p[X]$:
    \begin{equation} \label{eq:cofactor_free_correction}
        \Delta G_h(X) \;=\; E_h(X) \cdot C(X) \bmod G_1(X)
    \end{equation}
    and updating the approximation in $\mathbb{Z}_{p^{h+1}}[X]$:
    \begin{equation} \label{eq:cofactor_free_update}
        G_{h+1}(X) \;=\; G_h(X) \;+\; p^h \cdot \Delta G_h(X) \bmod p^{h+1}
    \end{equation}
    which satisfies $X^n - 1 \equiv 0 \pmod{G_{h+1}(X), \, p^{h+1}}$.
    
    \item[\textup{(3)}] The cached inverse $C(X)$ is independent of the precision depth $h$: it is computed once over $\mathbb{F}_p$ and reused identically at every layer $h = 1, 2, \dots, e-1$.
\end{enumerate}
Consequently, after the one-time precomputation of $C(X)$, each lifting iteration requires only a polynomial multiplication and modular reduction in $\mathbb{F}_p[X]/\langle G_1(X) \rangle$, costing $O(m^2)$ per layer---independent of both the global degree $n$ and the precision depth $e$.
\end{thm}

\begin{proof}
Let $G_1(X)$ be a monic irreducible factor of $X^n-1$ over $\mathbb{F}_p$, with cofactor $H(X) = (X^n-1)/G_1(X) \in \mathbb{F}_p[X]$. Suppose that after $h-1$ lifting steps we have obtained a monic polynomial $G_h(X) \in \mathbb{Z}_{p^h}[X]$ satisfying $G_h(X) \equiv G_1(X) \pmod p$ and a quotient $Q_h(X)$ such that $X^n - 1 \equiv Q_h(X) G_h(X) \pmod{p^h}$. We prove that the correction $\Delta G(X)$ required to extend $G_h(X)$ to $\mathbb{Z}_{p^{h+1}}[X]$ is determined solely by objects in $\mathbb{F}_p$.

Performing polynomial division over $\mathbb{Z}_{p^{h+1}}[X]$:
\begin{equation}
    X^n - 1 = Q_h(X) G_h(X) + p^h E(X)
\end{equation}
where the remainder $R(X) = p^h E(X)$ has all coefficients divisible by $p^h$, and $E(X) \equiv R(X)/p^h \pmod p$ is the reduced error polynomial in $\mathbb{F}_p[X]$.

Define the target lifted factors as $G_{h+1}(X) = G_h(X) + p^h \Delta G(X)$ and $Q_{h+1}(X) = Q_h(X) + p^h \Delta Q(X)$. The required congruence is:
\begin{equation} \label{eq:lift_target}
    X^n - 1 \equiv \big(G_h(X) + p^h \Delta G(X)\big)\big(Q_h(X) + p^h \Delta Q(X)\big) \pmod{p^{h+1}}
\end{equation}
Expanding \eqref{eq:lift_target} and observing that the cross term $p^{2h} \Delta G \Delta Q$ vanishes modulo $p^{h+1}$ (since $2h \geq h+1$ for $h \geq 1$):
\begin{equation} \label{eq:lift_expanded}
    X^n - 1 \equiv Q_h(X) G_h(X) + p^h \big(Q_h(X) \Delta G(X) + G_h(X) \Delta Q(X)\big) \pmod{p^{h+1}}
\end{equation}
Substituting $X^n - 1 = Q_h(X) G_h(X) + p^h E(X)$ and cancelling $Q_h G_h$ yields:
\begin{equation}
    p^h E(X) \equiv p^h \big(Q_h(X) \Delta G(X) + G_h(X) \Delta Q(X)\big) \pmod{p^{h+1}}
\end{equation}
Dividing by $p^h$ projects the constraint into $\mathbb{F}_p[X]$. Since $Q_h(X) \equiv H(X) \pmod p$ and $G_h(X) \equiv G_1(X) \pmod p$:
\begin{equation} \label{eq:projected_error}
    E(X) \equiv H(X) \Delta G(X) + G_1(X) \Delta Q(X) \pmod p
\end{equation}
Evaluating \eqref{eq:projected_error} modulo the base seed ideal $G_1(X)$ annihilates the second term:
\begin{equation}
    E(X) \equiv H(X) \Delta G(X) \pmod{G_1(X)}
\end{equation}
Because $\gcd(G_1(X), H(X)) = 1$ in $\mathbb{F}_p[X]$, the cofactor $H(X)$ is invertible in the quotient ring $\mathbb{F}_p[X]/\langle G_1(X) \rangle$, yielding the unique solution:
\begin{equation}
    \Delta G(X) \equiv E(X) \cdot [H(X)]^{-1} \pmod{G_1(X)}
\end{equation}
This derivation holds identically for every layer $h \geq 1$, completing the induction.
\end{proof}

\begin{rem}[On the ``Cofactor-Free'' Terminology]
The term ``cofactor-free'' refers to the \emph{lifting iterations}: once $C(X) = [H(X)]^{-1} \bmod G_1(X)$ is precomputed, the cofactor polynomial $H(X)$ itself (of degree $n - m$) is never accessed again. Each subsequent layer uses only $C(X)$ (of degree $< m$) and the current approximation $G_h(X)$. This is in contrast to classical Hensel lifting, which must maintain and update the full cofactor $H(X)$ at every precision layer.
\end{rem}

\section{Empirical Complexity Evaluation}

To validate the theoretical complexity bounds, we benchmark the Dickson Engine against two widely-used computer algebra systems: SymPy~\cite{sympy} (a pure-Python symbolic mathematics library implementing the Cantor--Zassenhaus algorithm) and SageMath~\cite{sagemath} (a Python interface to the highly optimized C/C++ libraries FLINT and Pari). All three contestants---the Dickson Engine, SymPy, and SageMath---are invoked through the same Python~3.12 runtime (CPython~3.12.10), ensuring a fair comparison of algorithmic efficiency rather than implementation-language advantages. In particular, SageMath's polynomial factorization dispatches to compiled C code (FLINT/Pari), while the Dickson Engine operates entirely in pure Python without any compiled-language backend. The software versions used are: SageMath~10.6, SymPy~1.14.0, and Python~3.12.10. All benchmarks were executed sequentially on an Intel Core Ultra~9 workstation running Ubuntu Linux. To ensure strict statistical rigor, every data point is derived from ten independent executions; the maximum and minimum are discarded, and the reported measurement is the arithmetic mean of the remaining eight runs. The complete open-source implementation is available at \url{https://github.com/Nothing256/Dickson-Engine}.

For the Dickson Engine, we report results under two initialization strategies: the \emph{Auto-Seeder}, which dynamically searches for a valid primitive polynomial seed at runtime via norm validation, and the \emph{Precomputed Seed} strategy, which eliminates initialization overhead by loading a known seed. For a rigorously fair comparison against general-purpose libraries (which compute from scratch), all comparative speedup multipliers highlighted in the text are calculated against the Auto-Seeder's total runtime, fully accounting for seed generation overhead.

\subsection{Base Field Factorization}

\begin{figure}[htbp]
    \centering
    \includegraphics[width=0.85\textwidth]{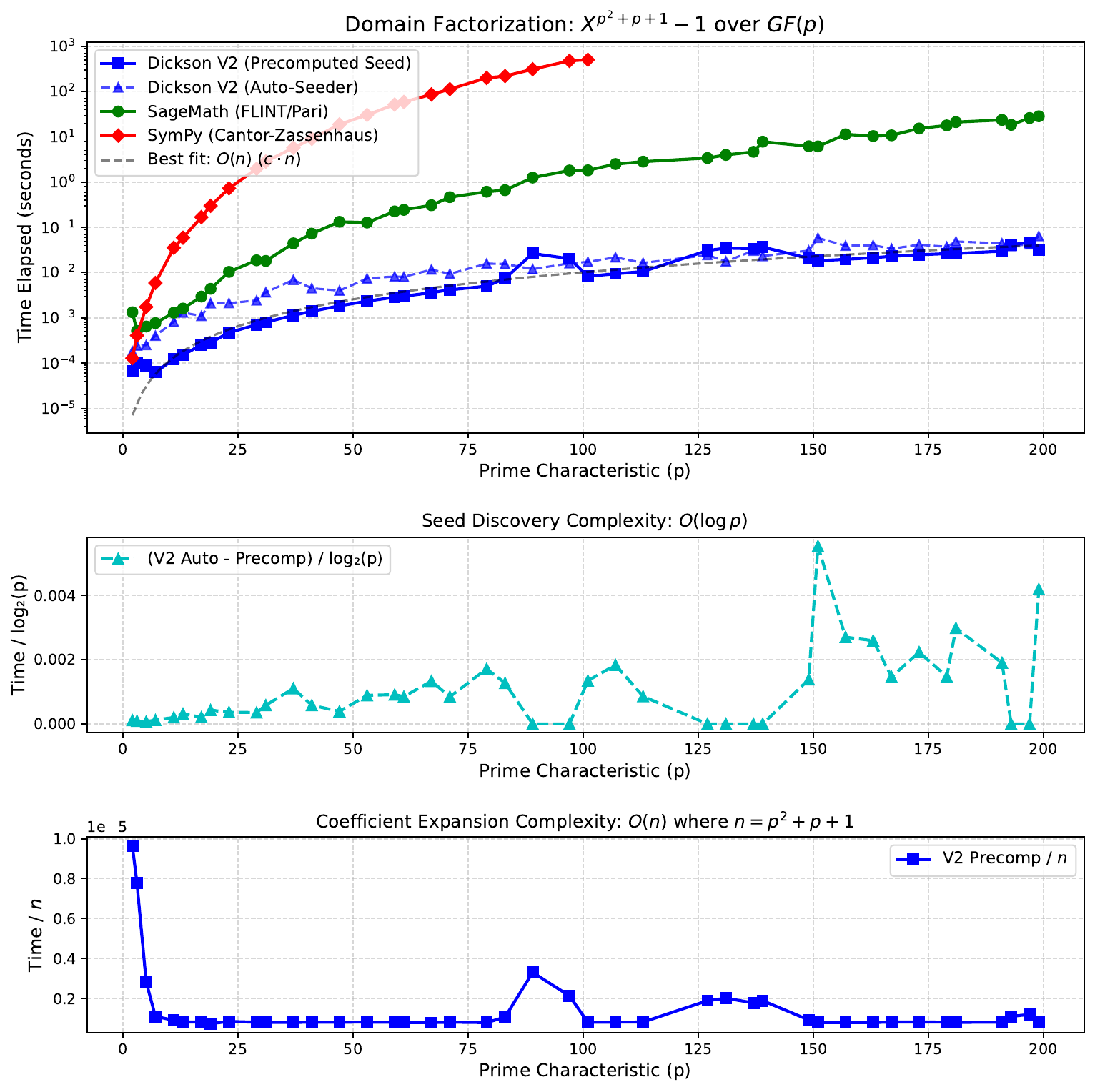}
    \caption{Domain factorization of $X^{p^2+p+1}-1$ over $GF(p)$ for primes $p = 2, \ldots, 199$ (10-iteration trimmed mean). SymPy (pure Python, Cantor--Zassenhaus) exhibits super-polynomial growth, reaching 508\,s at $p=101$ before timing out. SageMath (FLINT/Pari C backend) scales as $O(n^{1.5})$, reaching 28.6\,s at $p=199$. The Dickson V2 engine (pure Python) completes in 0.064\,s at $p=199$ using the Auto-Seeder---a $445\times$ speedup over SageMath (0.032\,s and $890\times$ with a precomputed seed). The lower panels isolate the independent cost components, normalizing the Auto-Seeder overhead by $\log_2 p$ and the coefficient expansion time by $n$, confirming the predicted $O(n + m^3 \log p)$ composite scaling.}
    \label{fig:domain_bench}
\end{figure}

As demonstrated in Figure~\ref{fig:domain_bench}, we benchmark base field factorization of $X^n - 1$ over $GF(p)$ with $n = p^2 + p + 1$, ensuring $\text{ord}_n(p) = 3$ uniformly across all test cases. SymPy's Cantor--Zassenhaus implementation exhibits super-polynomial complexity, rendering it impractical beyond $p \approx 100$ ($n \approx 10{,}000$); at $p = 101$, a single factorization requires 508\,s. SageMath, despite dispatching to the highly optimized C libraries FLINT and Pari, scales as $O(n^{1.5})$ and requires 28.6\,s at $p=199$ ($n = 39{,}801$). The Dickson V2 engine---implemented entirely in pure Python---completes the same factorization in 0.064\,s with the runtime Auto-Seeder (a $445\times$ speedup over SageMath), and in 0.032\,s with a precomputed seed (an $890\times$ speedup). This dramatic advantage demonstrates that the algorithmic paradigm (cofactor-free Hensel lifting with Dickson trace recurrence) dominates implementation-level optimizations: even without any compiled-language backend, the $O(n + m^3 \log p)$ composite complexity of the Dickson framework vastly outperforms the generic polynomial factorization algorithms employed by industrial-strength computer algebra systems. To explicitly verify this composite asymptotic behavior, the middle and lower panels of Figure~\ref{fig:domain_bench} isolate the independent cost components. By normalizing the pure seed discovery overhead by $\log_2 p$ and the coefficient expansion time by $n$, we obtain flat asymptotic trajectories. The visible variance in the normalized seed discovery overhead (middle panel) accurately reflects the probabilistic nature of the Auto-Seeder's randomized primitivity testing. This confirms that the framework's complexity strictly adheres to $O(n + m^3 \log p)$ (with $m=3$ here), rendering it practically immune to the exponential explosions that trap traditional algorithms. Furthermore, the raw execution time of the coefficient expansion phase (precomputed seed) tightly tracks the least-squares quadratic fit $O(n)$ displayed in the top panel, visually reinforcing its deterministic linear dependence on the number of generated factors.

\subsection{Ring Factorization}

\begin{figure}[htbp]
    \centering
    \includegraphics[width=0.85\textwidth]{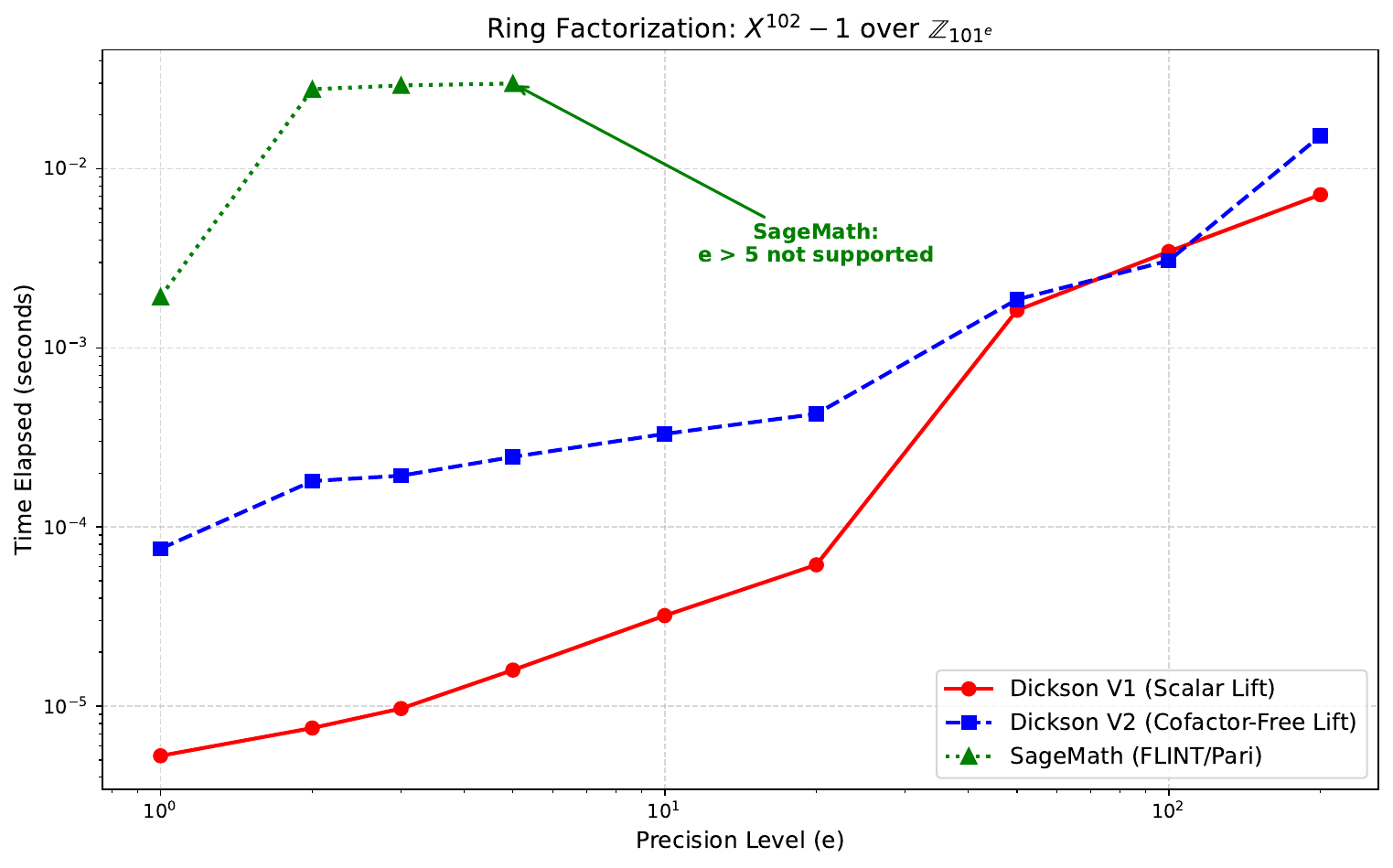}
    \caption{Ring factorization of $X^{102}-1$ over $\mathbb{Z}_{101^e}$ for precision depths $e = 1, \ldots, 200$ (10-iteration trimmed mean). SageMath raises \texttt{NotImplementedError} for $e > 5$, confirming that no general-purpose CAS currently supports polynomial factorization over $\mathbb{Z}_{p^e}$ with composite modulus. Both Dickson V1 (scalar lift) and V2 (cofactor-free lift) employ the Auto-Seeder and complete all precision levels; their comparable performance at $p=101$ reflects the moderate prime size where the asymptotic gap $O(e \cdot p)$ vs.\ $O(e \cdot m^2)$ has not yet fully separated.}
    \label{fig:ring_medium}
\end{figure}

\begin{figure}[htbp]
    \centering
    \includegraphics[width=0.85\textwidth]{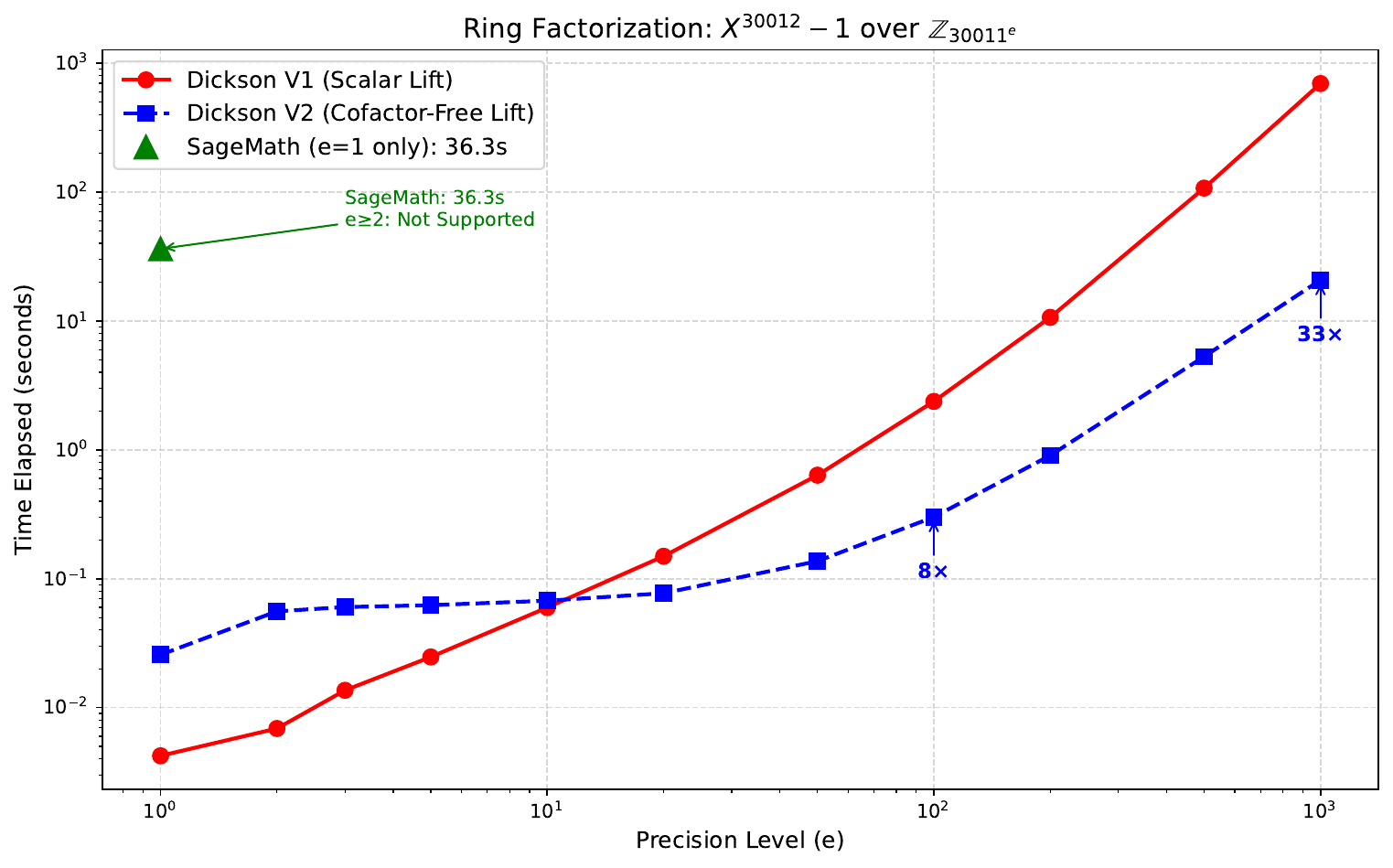}
    \caption{Ring factorization of $X^{30012}-1$ over $\mathbb{Z}_{30011^e}$ for precision depths $e = 1, \ldots, 1000$ (10-iteration trimmed mean). The V2 cofactor-free lift (using a precomputed seed to strictly isolate lifting performance) achieves a $33.5\times$ speedup over V1 at $e=1000$, confirming the $O(e \cdot m^2)$ vs.\ $O(e \cdot p)$ asymptotic separation. SageMath, shown at $e=1$ only (36.3\,s), raises \texttt{NotImplementedError} for all $e \geq 2$.}
    \label{fig:ring_large}
\end{figure}

Figures~\ref{fig:ring_medium} and~\ref{fig:ring_large} evaluate ring-level factorization over $\mathbb{Z}_{p^e}$, the regime uniquely addressed by our framework. Figure~\ref{fig:ring_medium} reveals that SageMath raises a \texttt{NotImplementedError} (``factorization of polynomials over rings with composite characteristic is not implemented'') for all $e > 5$, confirming that no general-purpose computer algebra system currently supports polynomial factorization over $\mathbb{Z}_{p^e}$ with composite modulus---precisely the regime addressed by the cofactor-free Hensel lift of Theorem~\ref{thm:cofactor_free}.

Figure~\ref{fig:ring_large} demonstrates the asymptotic separation between the V1 scalar lift (total $O(e \cdot p)$ algebraic operations) and the V2 cofactor-free lift (total $O(e \cdot m^2)$ algebraic operations) at full scale. At $p = 30{,}011$ and $e = 1000$, the V2 engine completes in 20.8\,s compared to 695.5\,s for V1---a $33.5\times$ speedup that validates the theoretical complexity bounds. SageMath, at $e=1$ alone, requires 36.3\,s for the same factorization (compared to V2's 0.026\,s), and is entirely unable to proceed beyond $e = 1$.

\section{Discussion}

As established in Section~1.1, the Ideal Derivation Modulo Principle formally subsumes the prior explicit factorization results of Fitzgerald--Yucas \cite{Fitzgerald2007} and Wang--Wang \cite{Wang2012}, which operated on individual cyclotomic polynomials $\Phi_d(X)$ over finite fields $\mathbb{F}_q$. Beyond reproducing these prior results as special cases, the present framework extends them in two fundamental directions: (i)~the Cofactor-Free Hensel Lift (Theorem~\ref{thm:cofactor_free}) enables precision elevation to $\mathbb{Z}_{p^e}$, which was unavailable in the purely finite-field setting; and (ii)~the MED-based coset dispatch (Section~4) and dual-track coefficient reconstruction (Proposition~\ref{prop:dual_track}) handle the complete factorization of $X^n - 1$ (encompassing all constituent $\Phi_d(X)$ simultaneously) for arbitrary $n$ with mixed-degree decompositions, removing the case-by-case analysis required by earlier approaches.

\section{Conclusion}

This paper establishes a complete, algorithmically explicit framework for factoring $X^n-1$ over $\mathbb{Z}_{p^e}$ for arbitrary coprime $(n, p)$. The Ideal Derivation Modulo Principle (Theorem~\ref{thm:ideal_modulo}) provides a unified characterization of all factor coefficients as roots of multivariate Dickson polynomial ideals, generalizing prior finite-field results \cite{Fitzgerald2007, Wang2012} from fixed dimensions to arbitrary cyclotomic orders. The cofactor-free Hensel lift (Theorem~\ref{thm:cofactor_free}) reduces ring-level factorization to a single seed elevation with a cached polynomial inverse, achieving $O(m^2)$ cost per precision layer independent of $e$. The dual-track coefficient reconstruction (Proposition~\ref{prop:dual_track}) ensures unconditional correctness across all characteristic regimes by combining Newton--Girard inversion ($p > m$) with quotient-ring Gaussian elimination ($p \leq m$). Ultimately, the separation of seed discovery, single-seed lifting, and independent coefficient expansion yields a grand total algebraic complexity of $O(n + m^3 \log p + e \cdot m^2)$ for the entire factorization process, guaranteeing computational immunity against the exponential scaling laws that bottleneck classical multivariate inversion.

The present framework requires the unramified condition $\gcd(n, p) = 1$. When $p \mid n$, the polynomial $X^n-1$ acquires repeated roots over $\mathbb{F}_p$, invalidating the coprime factorization underlying the Hensel lift. Extending the Dickson-MED architecture to this ramified regime constitutes a natural direction for future work.

\appendix
\section{Worked Example: $X^7-1$ over $\mathbb{Z}_8$}
\label{app:worked_example}

We trace Algorithm~\ref{alg:main} step by step for $n=7, p=2, e=3$.

\textbf{Stage 1: Seed Discovery.}
We seek a single irreducible factor of $X^7-1$ over $\mathbb{F}_2$ of degree $m = \text{ord}_7(2) = 3$. Testing $G_1(X) = X^3 + X + 1$: we verify $X^7 \equiv 1 \pmod{G_1, 2}$ and $X^1 \not\equiv 1 \pmod{G_1, 2}$ (the only prime divisor of $7$ is $7$ itself, and $7/7=1$). The seed passes.

\textbf{Stage 2: Single-Seed Hensel Lift.}
We compute $H(X) = (X^7-1)/G_1(X) = X^4 + X^3 + X^2 + 1$ over $\mathbb{F}_2$, and cache the inverse $C(X) = [H(X)]^{-1} \bmod G_1(X)$ over $\mathbb{F}_2$ once.

\emph{Layer $h=1$ ($\bmod\,4$):} Compute $X^7 \bmod G_{lift}(X)$ over $\mathbb{Z}_4$, extract the error $E(X)$, and update: $G_{lift}(X) = X^3 + 2X^2 + X + 3$.

\emph{Layer $h=2$ ($\bmod\,8$):} Repeat over $\mathbb{Z}_8$: $G_{lift}(X) = X^3 + 6X^2 + 5X + 7$.

The coefficient evolution of the single seed across precision layers:
\begin{align*}
    h=1 \ (\bmod\ 2): \quad & (1, 0, 1, 1) \\
    h=2 \ (\bmod\ 4): \quad & (1, 2, 1, 3) \\
    h=3 \ (\bmod\ 8): \quad & (1, 6, 5, 7) 
\end{align*}

\textbf{Stage 3: Trace Generation.}
From $G_{lift}(X) = X^3 + 6X^2 + 5X + 7$ over $\mathbb{Z}_8$, we extract the elementary symmetric polynomials: $A_1 = 2, A_2 = 5, A_3 = 1$ (after sign extraction). The Newton-Girard recurrence generates the global trace array:
\begin{equation*}
    S_1 = 2, \quad S_2 = 2, \quad S_3 = 5, \quad S_4 = 2, \quad S_5 = 5, \quad S_6 = 5
\end{equation*}

\textbf{Stage 4: Factor Reconstruction.}
Since $p = 2 \leq m = 3$, the Newton-Girard inversion encounters a zero-divisor ($k=2$ requires $2^{-1} \bmod 8$, which does not exist). The algorithm triggers the \emph{fallback track}: for coset representative $s = 3$, it computes $\beta = X^3 \bmod G_{lift}(X)$ and its powers in the quotient ring $\mathbb{Z}_8[X]/\langle G_{lift} \rangle$, then recovers the minimal polynomial via Gaussian elimination, yielding $X^3 + 3X^2 + 2X + 7$.

The trivial factor $(X+7) \equiv (X-1) \pmod{8}$ corresponds to the fixed coset $\{0\}$. Assembling all factors:
\begin{equation}
    X^7-1 \equiv (X+7)(X^3+6X^2+5X+7)(X^3+3X^2+2X+7) \pmod 8
\end{equation}
This demonstrates the complete pipeline: a single seed lifted once, a global trace array generated from it, and all remaining factors reconstructed without additional Hensel lifts.

\section{Derivation of the Remainder Coefficients}
\label{app:dickson_derivation}

We derive the explicit forms of $V_1(S, P)$ and $V_0(S, P)$ from the polynomial division of $X^n - 1$ by $G(S, P, X) = X^2 - SX + P$.

In the quotient ring $\mathbb{F}_p[X] / \langle X^2 - SX + P \rangle$, every element reduces to a linear polynomial $aX + b$. For each power $X^k$, we write:
\begin{equation}
    X^k \equiv a_k X + b_k \pmod{X^2 - SX + P}
\end{equation}
with initial conditions $X^0 = 1$ (giving $a_0 = 0, b_0 = 1$) and $X^1 = X$ (giving $a_1 = 1, b_1 = 0$). Applying the reduction $X^2 \equiv SX - P$ yields the recurrence:
\begin{align}
    a_k &= S \cdot a_{k-1} - P \cdot a_{k-2} \\
    b_k &= S \cdot b_{k-1} - P \cdot b_{k-2}
\end{align}

\textbf{Identifying $a_k$.} The recurrence $a_k = S a_{k-1} - P a_{k-2}$ with $a_0 = 0, a_1 = 1$ is precisely the defining recurrence of the Dickson polynomial of the second kind $E_k(S, P)$ (Definition~\ref{def:dickson_poly}). Therefore $a_k = E_{k-1}(S, P)$ for all $k \geq 1$.

\textbf{Identifying $b_k$.} The same recurrence with $b_0 = 1, b_1 = 0$ yields $b_2 = -P, b_3 = -PS, b_4 = -P(S^2 - P), \dots$ By induction, $b_k = -P \cdot E_{k-2}(S, P)$ for $k \geq 2$.

\textbf{Assembling the remainder.} Since $X^n - 1 \equiv a_n X + (b_n - 1)$, we obtain:
\begin{align}
    V_1(S, P) &= a_n = E_{n-1}(S, P) \\
    V_0(S, P) &= b_n - 1 = -P \cdot E_{n-2}(S, P) - 1
\end{align}
This establishes that both remainder coefficients are expressible in closed form via the Dickson recurrence.

\bibliographystyle{ieeetr}
\bibliography{ref}

\end{document}